\documentclass[journal]{IEEEtran}
\usepackage{amsmath, amssymb}
\usepackage[ruled,linesnumbered]{algorithm2e}
\usepackage{graphicx}
\usepackage{subfig}
\usepackage{mathtools}
\usepackage{enumitem}
\usepackage{color}
\SetKwInput{KwInput}{Input}
\SetKwInput{KwOutput}{Output}

\SetAlCapFnt{\bfseries}
\SetAlCapNameFnt{\normalfont}
\SetAlgoNlRelativeSize{0}

\SetAlgoBlockMarkers{begin}{end}
\SetKwFor{For}{for}{do}{end}
\SetKwIF{If}{ElseIf}{Else}{if}{then}{else if}{else}{end}
\SetKwFor{While}{while}{do}{end}

\newtheorem{lem}{Lemma}
\newtheorem{thm}{Theorem}

\newtheorem{remark}{Remark}

\begin{document}
\title{ Internal State-Based Policy Gradient Methods for Partially Observable Markov Potential Games
}

\author{Wonseok~Yang and~Thinh T. ~Doan
\thanks{This work was supported in part by the National Science Foundation (NSF) under  CAREER Award 2527059 and AFOSR Grant FA9550-25-1-0247.}
\thanks{Wonseok Yang and Thinh T. Doan are with the University of Texas at Austin, Austin,
TX, 78712 USA (e-mail:{wonseok.yang@utexas.edu, \,thinhdoan@utexas.edu})}
}

\maketitle

\begin{abstract}
This letter studies multi-agent reinforcement learning in partially observable Markov potential games. Solving this problem is challenging due to partial observability, decentralized information, and the curse of dimensionality. First, to address the first two challenges, we leverage the common information framework, which allows agents to act based on both shared and local information. Second, to ensure tractability, we study an internal state that compresses accumulated information, preventing it from growing unboundedly over time. We then implement an internal state-based natural policy gradient method to find Nash equilibria of the Markov potential game. Our main contribution is to establish a non-asymptotic convergence bound for this method. Our theoretical bound decomposes into two interpretable components: a statistical error term that also arises in standard Markov potential games, and an approximation error capturing the use of finite-state controllers. Finally, simulations across multiple partially observable environments demonstrate that the proposed method using finite-state controllers achieves consistent improvements in performance compared to the setting where only the current observation is used.
\end{abstract}

\section{Introduction}
Multi-agent systems have emerged as a powerful paradigm for tackling complex real-world tasks, with applications spanning search and rescue \cite{MLyu}, cooperative transportation \cite{BPandit}, and autonomous navigation \cite{JChoi}. Problems in these applications are naturally modeled as decentralized partially observable Markov decision processes (Dec-POMDPs), where agents compute policies from local observations. Exact solutions of Dec-POMDPs are computationally intractable due to partial observability, decentralized execution, and the exponential growth of the joint state–action space \cite{MKo}. The common information framework studied in \cite{ANayyar} mitigates these difficulties by having agents share a subset of local observations, enabling a reduction of Dec-POMDP to a centralized MDP over belief states, to which dynamic programming can be applied. This approach, however, requires the knowledge of the environment's transition dynamics, which are difficult or impossible to obtain in many applications. This motivates the use of data-driven approaches, such as reinforcement learning (RL), where policies are learned directly from data.

We study an RL framework for a structured subclass of Dec-POMDPs known as partially observable Markov potential games (POMPGs), which encompasses identical-reward team problems as a special case. In the fully observable setting, Markov potential games (MPGs) are well studied: prior work has established convergence of gradient-based dynamics to Nash equilibria (NE) \cite{SLeo,RZ} and characterized their iteration complexity \cite{RZhang,YSun}. Under partial observability, however, gradient-based methods remain largely unexplored. While the common information approach can reduce a POMPG to a belief-based MPG \cite{KHorak}, the resulting belief space is continuous and uncountable even for finite state spaces \cite{SSeuken}, rendering direct optimization intractable. Addressing this challenge requires efficient belief representations to design tractable gradient algorithms.

To this end, we develop a tractable policy gradient method for POMPGs that can scale to large multi-agent systems. Our approach builds on recent single-agent work \cite{SCayci,WMao}, where internal state representations compress observation histories to enable policy learning from finite memories. Extending this idea to multi-agent systems is non-trivial due to the decentralized structure; we therefore integrate the common information framework with finite internal state representations to yield a principled and tractable RL approach for POMPGs.

\noindent \textbf{Main Contributions.} We study internal state-based POMPGs, a structured subclass of POMGs in which each agent maintains finite internal states to compress its shared information history and local memory. By restricting policies to finite-state controllers, we prevent unbounded growth of common information and local memory, achieving computational tractability at the cost of exact optimality. Our focus is to develop a variant of the natural policy gradient (NPG) algorithm to optimize over the class of internal state-based policies, and provide theoretical guarantees on its performance in converging to a Nash equilibrium of the POMPG. Specifically, we establish a performance bound comprising an iteration complexity of $\mathcal{O}(1/\sqrt{T})$, where $T$ is the number of iteration, and an approximation error term arising from the use of finite-state controllers. Finally, we validate our theoretical findings through experiments across multiple RL environments, demonstrating the effectiveness of the proposed approach.
\section{Problem Formulation}
\subsection{POMPGs with information sharing}
We consider a partially observable Markov games (POMG) modeled as a tuple $\mathcal{G}=\{\,\mathcal{X},\{\mathcal{Y}_i\},\{\mathcal{U}_i\},P,\{\Omega_i\},\{r_i\}\}$, where $i\in \mathcal{N}=\{1,\cdots,n\}$ and $n$ is the number of agents. In this model, $\mathcal{X}$ is a finite set of states, and $\mathcal{Y}_i,\mathcal{U}_{i}$ denote a finite set of observations and actions at agent $i$, respectively. Let $\mathcal{U}=\mathcal{U}_1\times\cdots\mathcal{U}_n$ be the set of joint actions. Then, $P:\mathcal{X}\times\mathcal{U}\rightarrow\Delta\mathcal{X}$ is a state transition probability matrix, where $\Delta\mathcal{X}$ is a distribution over $\mathcal{X}$. 
Finally, $\Omega_i:\mathcal{X}\rightarrow \Delta \mathcal{Y}_i$ is a local probabilistic observation model and $r_i:\mathcal{X}\times\mathcal{U}\rightarrow\mathbb{R}$ denotes the local reward function at agent $i$.

Let $\mathcal{H}_i$ be the local information set available to agent $i$ and $\mathcal{H}=\cup_{i} \mathcal{H}_i$. A joint policy $\pi=(\{\mu_i^1\},\{\mu_i^2\},\cdots)$ is a sequence of mappings, where $\mu_i^k:\mathcal{H}_i^k\rightarrow\Delta_{\mathcal{U}_i}$ specifies the actions of agent $i$ at time $k$. Given a joint policy $\pi$ and an initial information $h^{0}$, let $V_{i}^{\pi}$ be the discounted value function of agent $i$ corresponding to $\beta \in (0,1)$
\begin{equation}
V_i^{\pi}(h^0)=\mathbb{E}_{\tau \sim \mathcal{P}_{h^0,T,\Omega}^{\pi}}\big[\sum_{k=0}^{\infty}\beta^k r_i(x^k,u^k) \,|\, H_i^0=h^0\big], \\
\label{eq:obj}
\end{equation}
and $J_i(\pi)=\mathbb{E}_{h^0 \sim \zeta}\big[V_i^{\pi}(h^0)\big]$ under some initial distribution $\zeta$. In the sequel, we write $\pi = (\pi_{i},\pi_{-i})$, where $\pi_{-i}$ denotes the collection of agents' policies but agent i's. Given $\pi_{-i}$, each agent $i$ seeks a policy $\pi^{\star}_{i}$ that maximizes its function $J_i(\pi_{i},\pi_{-i})$. In this setting, a common objective is to search for a Nash equilibrium (NE) $\pi^*=(\pi_1^*,\cdots,\pi_n^*)$ satisfying 
\begin{equation*}
J_i(\pi_i^*,\pi_{-i}^*) \geq J_i(\pi_i,\pi_{-i}^*),\quad \forall \pi_{i},\; i\in\mathcal{N}. 
\end{equation*}
At an NE, no agent has a unilateral incentive to deviate from its equilibrium, provided that all other agents play at their NEs.

In this paper, we will focus on a subclass of POMG, namely, partially observable Markov potential games (POMPG). In this setting, there exists a potential function $\phi:\mathcal{X}\times\mathcal{U}_i\times\mathcal{U}_{-i}\rightarrow\mathbb{R}$ such that the following equality holds for any agent $i$, and any pair of policies $(\pi'_i,\pi_{-i})$ and $(\pi_i,\pi_{-i})$:
\begin{equation}
J_i(\pi'_i,\pi_i)-J_i(\pi_i,\pi_i)=\Phi(\pi'_i,\pi_i)-\Phi(\pi_i,\pi_i),
\label{eq:POMPG}
\end{equation}
where the global potential function $\Phi$ is defined as
\begin{equation}
\Phi(\pi)\coloneqq \mathbb{E}\big[\sum_{k=0}^\infty \beta^k \phi\big(x^k,u^k\big)\,\big|\,H^0=h^0\big].
\label{eq:obj}
\end{equation}
with $\phi_{\min}\leq\phi(x,u)\leq\phi_{\max}$ for all $(x,u) \in \mathcal{X}\times\mathcal{U}$. \\
\indent
Finding an NE in the decentralized partially observable MDP (Dec-POMDP) setting is generally challenging and intractable, owing to each agent's partial and decentralized access to the global environment state. To address this, we adopt the common information sharing framework of \cite{ANayyar}, which transforms the Dec-POMDP into a centralized MDP. In this framework, each agent shares a portion of its local observations to a common memory accessible by all agents, and each agent then determines its NE policy using both the shared common information and its own local observation. We present this idea as follows.

Let $m_i^{k}\in\mathcal{M}_i$ be the local memory at time $k$ at agent $i$, i.e., $m_i^k\subset\{y_i^0,\cdots y_i^{k-1},u_i^0,\cdots,u_i^{k-1}\}$. The shared memory $c_{k}\in\mathcal{C}^{k}$ aggregates observations and actions broadcast by all agents,e.g., $c^k\subset \{y^0,\cdots,y^{k-1},u^0,\cdots,u^{k-1}\}$. At time $k$, each agent obtains local observation $y_i^k$ and selects action $u_i^k$ based on its available information $h_i^k=(y_i^k,m_i^k,c^k)$. It then transmits a portion $z_i^k\subset\{m_i^k,y_i^k,u_i^k\}$ to update the shared memory $c^{k}$.

As established in \cite{ANayyar}, the common information approach reduces the Dec-POMDP to a standard POMDP, which can in turn be solved by reformulating it as an equivalent MDP over the belief state space — the space of probability distributions over $\mathcal{X}$. However, this approach faces two fundamental challenges: (1) the belief state space is uncountable even when 
$\mathcal{X}$ is finite; and (2) computing belief states requires to store the entire history of agents' observations and the knowledge of the transition kernel $\mathcal{P}$, i.e., a full model of the environment. We address both challenges by introducing finite internal states and implementable approximation to the belief states. Building on this, we will develop a reinforcement learning algorithm in which each agent learns a finite-state NE policy directly from its local observations, without requiring knowledge of $\mathcal{P}$.

\subsection{Finite Internal State-Based Policies}
We introduce the finite internal state representation for both shared information and local observations at the agents. Let  $\mathcal{W}$ be a finite set and $w^k\in\mathcal{W}$ represent a compression of the shared information history up to time $k$. The shared internal state is updated as $w^{k+1}=\varphi(w^k,z^k,u^k)$ for some chosen function $\varphi:\mathcal{W}\times\mathcal{Z}\times\mathcal{U}\rightarrow\mathcal{W}$. Similarly, we denote by $\mathcal{L}_i$ a finite set and $l_i^k\in\mathcal{L}_i$ a compression of the local memory at time $k$. The local compressed memory $l_i$ is updated as   $\l_i^{k+1}=\chi_i(l_i^k,y_i^k,u_i^k,z_i^k),$ where $\xi_{i}:\mathcal{L}_{i}\times Y_{i}\times U_i\times \mathcal{Z}_{i}\rightarrow\mathcal{L}_{i}$ is a prior chosen function.  Examples of these functions include a finite window of past observation–action histories and feature representations produced by recurrent neural network (RNN)–based models \cite{SCayciRNN}. 

Given the internal state-based system, we will consider the class of finite-state controllers (FSC), denoted as $\Pi_{i}$, where each agent takes action $u_i^k$ based on the shared internal state $w^k$, the local internal state $l_i^k$, and the local observation $y_i^k$. The objective of the agents is to search a NE policy over $\Pi= \Pi_{1}\times\ldots\Pi_{n}$, i.e., they seek a $\pi^{\star}\in\Pi$ satisfying
    \begin{equation}
J_i(\pi_i^*,\pi_{-i}^*) \geq J_i(\pi_i,\pi_{-i}^*),\quad \forall \pi_{i}\in\Pi_{i},\; i\in\mathcal{N}.\label{eq:NE-finite-controller} 
\end{equation}
\section{Internal State Natural Policy Gradient}
We propose an internal state-based natural policy gradient (NPG) method to find the finite-state NE policy of the POMPG defined in \eqref{eq:NE-finite-controller}. The NPG is a gradient-based algorithm that updates policy parameters while accounting for the geometry of the policy distribution space rather than the Euclidean geometry of the parameter space \cite{AAgarwal}. 
Let each policy $\pi_i$ parameterized by $\theta_i\in\mathbb{R}^{|\hat{h}_i||u_i|}$ and our focus is to study the tabular softmax policies defined as
$$\pi_{\theta_i}(u_i|\hat{h}_i)=\frac{\exp(\theta_{\hat{h}_i,u_i})}{\sum_{u'_i}\exp(\theta_{\hat{h}_i,u'_i})}\cdot$$ 
Let $(w^0,l^0,y^0)\sim\xi\in\Delta_{\mathcal{W}\times\mathcal{L}\times\mathcal{Y}}$ denote initial approximate information. The initial information available to agent $i$ is $(w^0,l_i^0,y_i^0)$, whose distribution is marginal $\xi_i\in\Delta_{\mathcal{W}\times\mathcal{L}_i\times\mathcal{Y}_i}$. Further, we define the discounted visitation distribution conditioned on $(w^0,l_i^0,y_i^0)$ under policy $\pi$ as 
\begin{equation*}
\begin{aligned}
&d_{(w^0,l_i^0,y_i^0)}^{\pi}(w,l,y) \\
&=(1-\beta)\sum_{k=0}^{\infty}\beta^k\text{Pr}^{\pi}(W^k=w,L_i^k=l_i,Y_i^k=y_i|w^0,l_i^0,y_i^0)
\end{aligned}
\end{equation*}
for any policy $\pi$ and $(w,l,y)$. Also, let the marginal discounted visitation distribution over information available to agent $i$  be $d_{(w^0,l_i^0,y_i^0)}^{\pi}(w,l_i,y_i)=\sum_{l_{-i}^0,y_{-i}^0}d_{(w^0,l_i^0,y_i^0)}^{\pi}(w,l,y)$. For $\xi_i$, let $d_{\xi_i}^{\pi}(w,l_i,y_i)=\mathbb{E}_{(w^0,l_i^0,y_i^0)\sim\xi_i}\big[d_{(w^0,l_i^0,y_i^0)}^{\pi}(w,l_i,y_i)\big]$. 
For notational convenience, let $\hat{h}_i=(w,l_i,y_i)$ and $\hat{h}=(w,l,y)$.

The NPG algorithm updates the policy parameter as 
\begin{equation}
\theta_i^{t+1} =\theta_i^{t}+\eta\,F_i(\theta_i^t)^{\dagger}\nabla_{\theta_i}J_i(\pi_{\theta}^t),
\end{equation}
where $F_i(\theta_i^t)$ is the Fisher information matrix
\begin{equation*}
\hspace{-0.1cm}F_i(\theta_i^t)=\mathbb{E}_{\hat{h}_i\sim d_{\xi_i}^{\pi},u_i\sim \pi_{\theta_i}^t}\big[ \nabla_{\theta_i}\log\pi_{\theta_i}^t(u_i|\hat{h}_i)\nabla_{\theta_i}^{\top}\log\pi_{\theta_i}^t(u_i|\hat{h}_i) \big]
\end{equation*}
and $F_{i}(\theta)^\dagger$ denote its Moore-Penrose inverse, respectively. Under softmax parameterization, the NPG algorithm updates the parameter $\theta_{\hat{h}_i,u_i}$ and the policy $\pi_i$ as
\begin{align}
\begin{aligned}
\theta_{\hat{h}_i,u_i}^{t+1}&=\theta_{\hat{h}_i,u_i}^t+\cfrac{\eta}{1-\beta}\,A_i^{\pi^t}(\hat{h}_i,u_i), \\
\pi_i^{t+1}(u_i|\hat{h}_i)&= \pi_i^t(u_i|\hat{h}_i)\exp\left( \cfrac{\eta A_i^{\pi^t}(\hat{h}_i,u_i)}{1-\beta} \right)\Big/g_i^t(\hat{h}_i),
\end{aligned}
\label{eq:NPGpi}
\end{align}
where $g_i^t(\hat{h}_i)= \sum_{u_i}\pi_i^{t}(u_i|\hat{h}_i)\exp\big(\eta A_i^{\pi^t}(\hat{h}_i,u_i)/(1-\beta)\big)
$
is a normalized constant \cite{RZhang}. For each iteration, all agents update their own policy $\pi_i$ at the same time, and this algorithm requires the access to $A_i^{\pi^t}$ for synchronous update. The advantage function can be obtained through Monte-Carlo or TD-learning, but we assume that we know the exact advantage function. Details of the NPG algorithm is described in Algorithm 1.
\begin{algorithm}[t]
\caption{Internal State NPG Algorithm}
\label{alg:npg}
\DontPrintSemicolon

\KwInput{Initial policy $\pi^0$ and learning rate $\eta$}
    \For{$t = 0$ \KwTo $T-1$}{
        All agents simultaneously:\;
        \Indp
        Update policy using \eqref{eq:NPGpi}\;
        \Indm
    }
\KwOutput{Optimized policy $\pi^T$}
\end{algorithm}

\section{Main Results}
We present our main theoretical result in Theorem \ref{main_theorem}, where we provide an upper bound to characterize the convergence properties of the proposed internal state NPG method. Let $b(\cdot|c)$ and $b(\cdot|w)$ be the belief states defined over $(x,y,m)$ conditioned on the common information and shared internal state, respectively. In addition, let $d_{b}$ be the total variation distance between these two distributions
\begin{align}
    d_{b}&=\max_{w}\|b(\cdot\,|\,c)-b(\cdot\,|\,w)\|_{\text{TV}}.
\end{align}
\indent
Our theoretical bound will depend on this TV distance as the agents can only have access to the shared internal state. For our result, we will assume that $\inf_{\pi}\min_{i}\min_{\hat{h}_i}d_{\xi_i}^{\pi}(\hat{h}_i)>0$, which basically guarantees sufficient exploration of the initial policies $\rho_{i}$ and has been used extensively in the literature \cite{AAgarwal}. Finally, we define the NE-gap of agent $i$ as $\text{NE-gap}_i(\pi)=\sup_{\pi'_i}J_i(\pi'_i,\pi_{-i})-J_i(\pi_i,\pi_{-i}),$ and NE-gap as $\text{NE-gap}(\pi)=\max_i\text{NE-gap}_i(\pi)$.

\begin{thm}\label{main_theorem}
  Let $\eta = (1-\beta)^2/\big(2n\phi_{\max}\big)$. Then the policies $\{\pi^{t}\}$ generated by Algorithm \ref{alg:npg} satisfy
\begin{equation}
\cfrac{1}{T}\sum_{t=0}^{T-1}\,\text{NE-gap}(\pi^t)
\leq \mathcal{O}\big(\sqrt{\frac{n}{aT}}\big) + \varepsilon_{\text{FSC}},\label{main_theorem:ineq}
\end{equation}
where
\[
\begin{aligned}
\varepsilon_{\text{FSC}}&=\cfrac{2\sqrt{2}\phi_{\max}}{1-\beta}\sqrt{\big(d_b^{2}+\cfrac{3Mn}{a(1-\beta)}\,d_b\big)}, \\
a&=\min_{\hat{h}_i}\sum_{u_i^*}\pi_i(u_i^*\,|\hat{h}_i), \,\, u_i^*= \arg\max_{u_i} Q_i^{\pi^t}(\hat{h}_i,u_i).
\end{aligned}
\]
\end{thm}

\indent
\textit{Remark 1}: The proposed algorithm exhibits performance bound on the averaged NE gap consisting of $\mathcal{O}(1/\sqrt{T})$ term and the additional error term. If the state becomes fully observable, then the proposed performance bound can be reduced to $\mathcal{O}(1/\sqrt{T})$, which aligns with the result of \cite{RZhang}. Furthermore, the performance bound with $n=1$ includes $\mathcal{O}(1/\sqrt{T})$, analogous to the bound in \cite{SCayci}. 

The theorem further indicates that the presence of partial observability and use of FSC induces $\varepsilon_{\text{FSC}}$ which quantifies the difference in estimation of state between common information and internal state that approximates the common information. This implies that expressive internal-state representation reduces this gap, leading to improve the performance.  

The performance bound depends on the $1/\sqrt{a}$ that can be possible to become large, as similar to MPGs setting. A small value of $a$ represents the case where the policy far from the NE has difficulty in updating due to the small gradient. This happens when the initial policy is uniform initialization with the large size of the internal state. 

\subsection{Proof of Theorem \ref{main_theorem}}
We next present the analysis to derive the result in Theorem \ref{main_theorem}. We first present the following intermediate results that are necessary for our analysis. For convenience, we present their proofs in the Appendix. 

\begin{lem}
    Given two policies $\pi',\pi$ we have 
\begin{equation}
\begin{aligned}
\Phi(\pi')-\Phi(\pi) &\leq \cfrac{1}{1-\beta}\,\mathbb{E}_{\hat{h}_i,u}\big[A_{\phi}^\pi(
\hat{h}_i,u)\big] \\
&+ \cfrac{2\phi_{\max}}{1-\beta}\,\mathbb{E}_{\hat{h}_i,u}\big[\,\|b(\cdot\,|\,c)-b(\cdot\,|\,w)\|_{\text{TV}}\,\big],
\label{eq:lemma1}
\end{aligned}
\end{equation}
where $(\hat{h}_i,u)\sim (d_{\xi_i}^{\pi'}\otimes\pi')(\cdot)$ are sampled from the marginal state-action distribution induced by policy $\pi'$, and $(d_{\xi_i}^{\pi'}\otimes\pi')(\hat{h}_i,u)=\sum_{l_{-i},y_{-i}}d_{\xi_i}^{\pi'}(\hat{h})\pi'(u|\hat{h})$.
\end{lem}
\begin{remark}
Lemma $1$, a variant of the well-known performance difference lemma in MDP,   consists of two terms: one associated with the advantage function and the other capturing the error induced by partial observability and the use of finite-state controllers. In fully observable settings, the latter term vanishes, and the result reduces to the standard performance difference bound for fully observable MPG. In contrast, for internal state–based POMPG, this additional term influences resulting finite-time convergence bounds.
\end{remark}
\begin{lem}
    For any $\pi$, we have
\begin{equation}
\text{NE-GAP}(\pi) \leq \cfrac{1}{1-\beta}\,\max_{i}\,\max_{\hat{h}_i,u_i}A_{i}^\pi(\hat{h}_i,u_i) + \cfrac{2d_b\phi_{\max}}{1-\beta}\cdot
\label{eq:lemma2}
\end{equation}
\end{lem}
\begin{lem}
The sequence $\{\pi^{t}\}$ generated by \eqref{eq:NPGpi} satisfies
\begin{equation}
\begin{aligned}
&\cfrac{1}{1-\beta}\,\mathbb{E}_{\hat{h}_i,u}\big[ A_{\phi}^{\pi^t}(\hat{h}_i,u)\big] \\
&\geq \kappa\sum_{\rho_i}d_{\xi_i}^{\pi^{t+1}}(\hat{h}_i)\,\sum_{j=1}^{n} \text{KL}\big(\pi_i^{t+1}(\cdot\,|\,\hat{h}_j)\,\|\,\pi_i^{t}(\cdot\,|\,\hat{h}_j)\big) \\
&\quad + \cfrac{1}{\eta}\,\sum_{\rho_i}d_{\xi_i}^{\pi^{t+1}}(\hat{h}_i)\,\sum_{i=1}^{n}\,\log \big( g_i^t(\hat{h}_i)\big),
\label{eq:lemma3}
\end{aligned}
\end{equation}
where $\kappa=\cfrac{1}{\eta}-\cfrac{2n\phi_{\max}}{(1-\beta)^2}\cdot$
\end{lem}
\begin{lem}
  The sequence $\{\pi^{t}\}$ generated by \eqref{eq:NPGpi} satisfies
\begin{equation}
\begin{aligned}
&\sum_{\hat{h}_i}d_{\xi_i}^{\pi^{t+1}}(\hat{h}_i)\,\sum_{i=1}^{n}\log\big(g_i^t(\hat{h}_i)\big) \\
&\geq \cfrac{a \eta^2}{3M}\,\big(\text{NE-gap}(\pi^t)-\cfrac{2d_b\phi_{max}}{1-\beta}\big)^2,
\label{eq:lemma4}
\end{aligned}
\end{equation}
where $M=\sup_{\pi} \max_{\hat{h}_i}\frac{1}{d_{\xi_i}^{\pi}(\hat{h}_i)}$.
\end{lem}

\noindent \textbf{Proof of Theorem \ref{main_theorem}}: Using \eqref{eq:lemma1}--\eqref{eq:lemma4} with $\eta = \frac{(1-\beta)^2}{2n\phi_{\max}}$ yields
\[
\begin{aligned}
&\Phi(\pi^{t+1}) - \Phi(\pi^t) + \cfrac{2d_b\phi_{\max}}{1-\beta} \\
&\geq \cfrac{a\eta}{3M}\,\big(\text{NE-gap}(\pi^t)-\cfrac{2d_b\phi_{max}}{1-\beta}\big)^2,
\end{aligned}
\]
which when summing up over $t$ gives
\[
\begin{aligned}
&\cfrac{1}{T}\sum_{t=0}^{T-1}\big(\text{NE-gap}(\pi^t)-\cfrac{2d_b\phi_{max}}{1-\beta}\big)^2 \\ &\leq \cfrac{3M}{a\eta T}\big(\big(\Phi(\pi^T)-\Phi(\pi^0)\big)+\cfrac{2d_bT\phi_{\max}}{1-\beta}\big) \\
&\leq \cfrac{3M\phi_{\max}}{a(1-\beta)\eta T} + \cfrac{6d_bM\phi_{\max}}{a(1-\beta)\eta}\cdot
\end{aligned}
\]
By using the relation $\lambda_1^2\leq2\lambda_2^2+2(\lambda_1-\lambda_2)^2$, where $\lambda_1=\text{NE-gap}(\pi^t)$ and $\lambda_2=\cfrac{2d_b\phi_{max}}{1-\beta}$, we obtain
\[
\begin{aligned}
\text{NE-gap}(\pi^t)^2 \leq \cfrac{8d_b^{2}\phi_{\max}^2}{(1-\beta)^2}+2\big(\text{NE-gap}(\pi^t)-\cfrac{2d_b\phi_{max}}{1-\beta}\big)^2,
\end{aligned}
\]
which when summing over $t$ and using the relation above gives
\[
\begin{aligned}
&\cfrac{1}{T}\sum_{t=0}^{T-1}\,\text{NE-gap}(\pi^t)^2  \\ &\leq \cfrac{8d_b^{2}\phi_{\max}^2}{(1-\beta)^2}+2\big(\cfrac{3M\phi_{\max}}{a(1-\beta)\eta T} + \cfrac{6d_bM\phi_{\max}}{a(1-\beta)\eta}\big) \\
&=\big(\cfrac{12Mn\phi_{\max}^2}{a(1-\beta)^3}\big)\cfrac{1}{T} + \cfrac{8\phi_{\max}^2}{(1-\beta)^2}\big(d_b^{2}+\cfrac{Mn}{a(1-\beta)}d_b\big).
\end{aligned}
\]
\indent
Applying the Jensen's inequality to the preceding relation immediately gives Eq. \eqref{main_theorem:ineq}. This concludes our proof.
\section{Simulation}

\begin{figure}[t]
    \centering
    \includegraphics[width=\columnwidth]{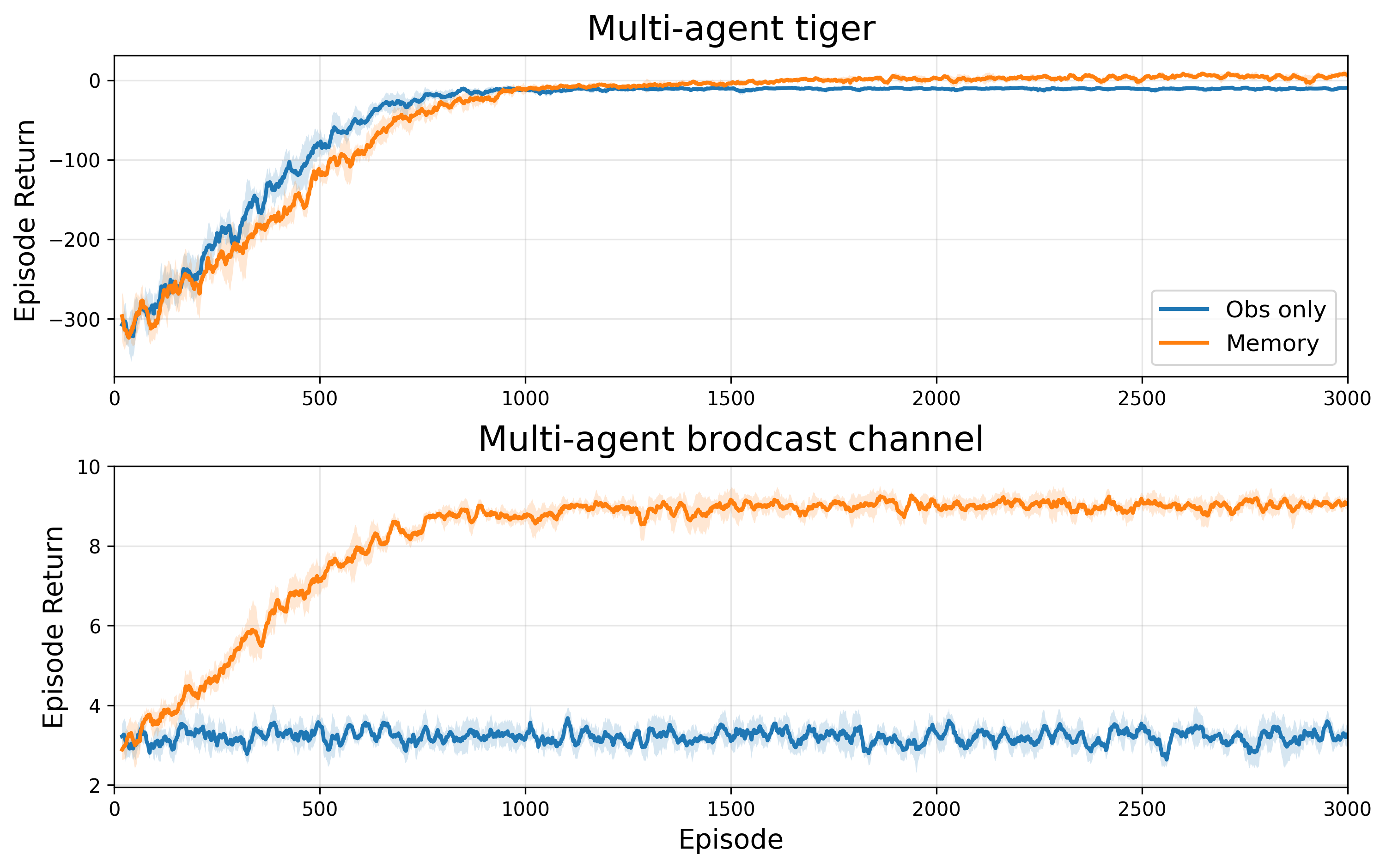}
    \caption{Learning curve of ITRPO for MATiger and MABC}
    \label{fig:training_curve}
\end{figure}
\begin{figure}[t]
    \centering
    \includegraphics[width=\columnwidth]{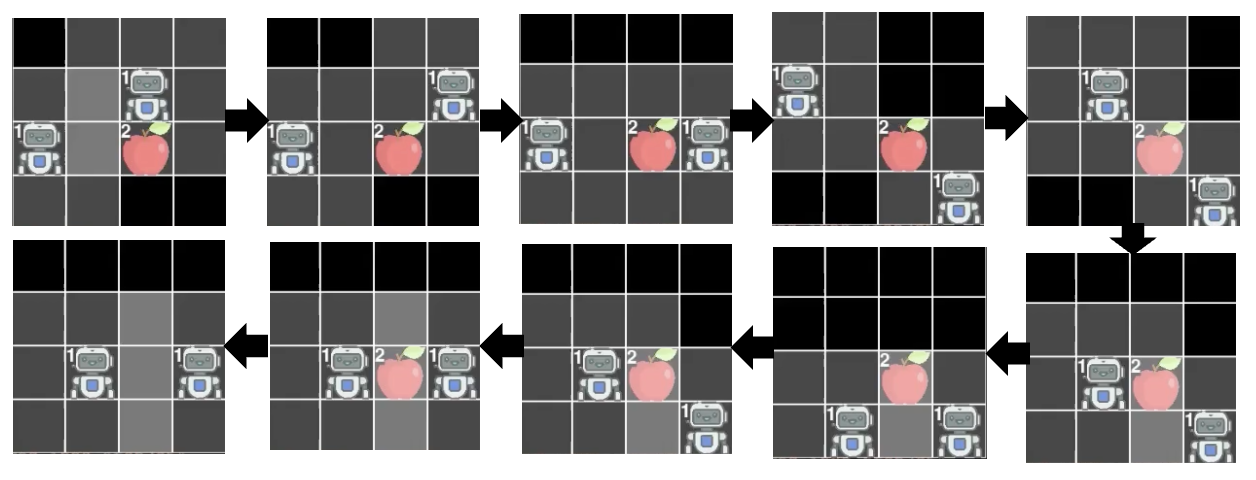}
    \caption{One run of the learned policy on LBF environment}
    \label{fig:training_curve}
\end{figure}
We evaluate the proposed NPG algorithm on partially observable common-reward games across three environments from the POSGGym benchmark \cite{JSch}: Multi-Agent Tiger (MATiger), Multi-Access Broadcast Channel (MABC), and Level-Based Foraging (LBF). 

In MATiger, two agents face two doors concealing a tiger and a treasure, respectively. At each step, each agent independently chooses to open the left door, open the right door, or listen. Opening the treasure door yields a reward of +10, opening the tiger door incurs a penalty of -100, and listening yields -1. Agents receive noisy observations of the tiger's location and the doors opened, with the objective of maximizing cumulative team reward over a 10-step episode.

In MABC, two nodes share a communication channel that supports only one successful transmission at a time. Each node independently chooses to transmit or remain idle, receiving a reward of $1$
if exactly one node transmits and $0$
 otherwise, along with noisy collision observations. Agents share the objective of maximizing channel throughput over a $10$-step episode.

In LBF, two agents navigate a $4\times 4$ grid world containing one food item, with each agent's observation limited to a sight range of 1. Agents may move within the grid or execute a lift action to collect adjacent food, upon which a reward is received and the episode terminates. The maximum episode length is set to 50.

In our experiments, neural network (NN)-based policy and value function are used for implementability. We use independent Trust Region Policy Optimization (ITRPO) \cite{JSchulman} with parameter sharing, which is a variant of NPG. The finite window of memory is selected as the internal state representation, and the finite-memory policy is used for finite-state controller. In addition, we set the local internal state at time $t$ as $l_i^t=(y_i^{t-t_w:t-1},u_i^{t-t_w:t-1})$ and the shared internal state as $w^t=y_{-i}^{t-t_w:t-1}$, where $t_w$ is the length of memory. We let $t_w=2$ for all experiments. 

\noindent\textbf{Main Observations.} The simulation results yield several key insights. As shown in Figure 1, the finite-memory policy consistently achieves higher returns than the reactive policy across both MATiger and MABC environments. Notably, the reactive policy fails to learn in MABC, underscoring that memory is essential for effective decision-making under partial observability. Figure 2 further illustrates that the finite-memory policy learned in the LBF environment produces coordinated agent behavior—agents align their movements and successfully complete the task—demonstrating that internal states naturally facilitate structured cooperation. Beyond performance improvements, these results reveal that internal state representations serve as a structural mechanism that implicitly approximates belief states without requiring explicit belief updates, thereby enabling tractable learning. The learning curves exhibit consistent improvement before saturating at large iteration counts, in agreement with the convergence properties established in Theorem 1. Finally, the performance gap between the finite-memory and reactive policies reflects the approximation error introduced by finite-state controllers, suggesting that richer internal state representations could further reduce this gap.
\section{Appendix}
\subsection{Proof of Lemma 1}
Let $V_{\phi}^{\pi}(\hat{h}_i^0)= \mathbb{E}^{\pi'}\big[\sum_{k=0}^{\infty}\beta^k\phi(x^k,u^k)|\hat{h}_i^0\big]$. The performance difference can be expressed as
\begin{equation}
\begin{aligned}
& V_{\phi}^{\pi'}(\hat{h}_i^0)-V_{\phi}^{\pi}(\hat{h}_i^0) \\
&= \mathbb{E}^{\pi'}\big[\sum_{k=0}^{\infty}\beta^k\big(\phi(x^k,u^k)+\beta V_{\phi_0}^{\pi}(s^{k+1},\hat{h}_i^{k+1}) \\
&\qquad\qquad\qquad\qquad -V_{\phi}^{\pi}(\hat{h}_i^k)\big)|\hat{h}_i^0\big] \\ 
&\quad+ \beta\mathbb{E}^{\pi'}\big[\sum_{k=0}^{\infty}\beta^k\big(V_{\phi}^{\pi}(\hat{h}_i^{k+1})-V_{\phi_0}^{\pi}(s^{k+1},\hat{h}_i^{k+1})\big)|\hat{h}_i^0\big],
\label{eq:pfd}
\end{aligned}
\end{equation}
where $V_{\phi_0}^\pi(s^0,\hat{h}_i^0) = \mathbb{E}^{\pi'}\big[\sum_{k=0}^{\infty}\beta^k\phi(x^k,u^k)|s^0,\hat{h}_i^0\big],$ and $s^k=(x^k,y^k,m^k).$ Applying the law of total expectation to the first term in \eqref{eq:pfd}, we obtain
\begin{equation}
\begin{aligned}
&\mathbb{E}^{\pi'}\big[\phi(x^k,u^k)+\beta V_{\phi_0}^{\pi}(s^{k+1},\hat{h}_i^{k+1})-V_{\phi}^{\pi}(\hat{h}_i^k)|\hat{h}_i^0\big] \\
&= \mathbb{E}^{\pi'}\big[\mathbb{E}^{\pi'}\big[\phi(x^k,u^k)+\beta V_{\phi_0}^{\pi}(s^{k+1},\hat{h}_i^{k+1})|c^k,h^k\big] \\
&\qquad\qquad-V_{\phi}^{\pi}(\hat{h}_i^k)|\hat{h}_i^0\big].
\label{eq:pfd1}
\end{aligned}
\end{equation}
\indent
Rewrite the conditional expectation in \eqref{eq:pfd1} with $d_b^k=\|b(\cdot\,|\,c^k)-b(\cdot\,|\,w^k)\|_{\text{TV}}$ and  $Q_{\phi_0}^\pi(\cdot)\leq \phi_{\max}/(1-\beta)$ where $Q_{\phi_0}^\pi(s^0,\hat{h}_i^0,u^0) = \mathbb{E}^{\pi'}\big[\sum_{k=0}^{\infty}\beta^k\phi(x^k,u^k)|s^0,\hat{h}_i^0,u^0\big]$, then we obtain
\begin{equation}
\begin{aligned}
&\mathbb{E}^{\pi'}\big[\phi(x^k,u^k)+\beta V_{\phi_0}^{\pi}(s^{k+1},\hat{h}_i^{k+1})|c^k,\hat{h}^k\big] \\
&= \sum_{s^k,u^k}b(s^k|c^k)\,\pi'(u^k|\hat{h}^k)\,Q_{\phi_0}^{\pi}(s^k,\hat{h}_i^k,u^k) \\
&= \sum_{s^k,u^k}\big(b(s^k|c^k)-b(s^k|w^k)\big)\pi'(u^k|\hat{h}^k)Q_{\phi_0}^{\pi}(s^k,\hat{h}_i^k,u^k) \\
&\qquad +\sum_{u^k}\pi'(u^k|\hat{h}^k)\,Q_{\phi}^{\pi}(\hat{h}_i^k,u^k). \\
&\leq \cfrac{\phi_{\max}}{1-\beta}\,d_b^k + Q_{\phi}^{\pi}(\hat{h}_i^k,u^k).
\label{eq:lemma1o}
\end{aligned}
\end{equation}
\indent
Substituting \eqref{eq:lemma1o} into \eqref{eq:pfd1}, we obtain the first term in \eqref{eq:pfd} as
\begin{equation}
\begin{aligned}
&\sum_{k=0}^{\infty}\beta^k\,\mathbb{E}^{\pi'}\big[\mathbb{E}^{\pi'}\big[\phi(s^k,u^k)+\beta V_{\phi_0}^{\pi}(s^{k+1},\hat{h}_i^{k+1})\big|c^k,\hat{h}^k\big] \\
&\qquad\qquad\qquad\quad-V_{\phi}^{\pi}(\hat{h}_i^k)\big|\hat{h}_i^0\big] \\
&\leq \sum_{k=0}^{\infty}\beta^k\,\mathbb{E}^{\pi'}\bigg[\cfrac{\phi_{\max}}{1-\beta}d_b^k + Q_{\phi}^{\pi}(\hat{h}_i^k,u^k) -V_{\phi}^{\pi}(\hat{h}_i^k)\,\bigg|\,\hat{h}_i^0\bigg] \\
&=\mathbb{E}^{\pi'}\big[\sum_{k=0}^{\infty}\beta^kA_{\phi}^\pi(\hat{h}_i^k,u^k)|\hat{h}_i^0\bigg] + \cfrac{\phi_{\max}}{1-\beta}\mathbb{E}^{\pi'}\big[\sum_{k=0}^{\infty}\beta^k d_b^k|\hat{h}_i^0\big].
\label{eq:pfd11}
\end{aligned}
\end{equation}
\indent
A similar argument can be applied to the second term, yielding
\begin{equation}
\begin{aligned}
&\beta\mathbb{E}^{\pi'}\big[\sum_{k=0}^{\infty}\beta^k\big(V_{\phi}^{\pi}(\hat{h}_i^{k+1})-V_{\phi_0}^{\pi}(s^{k+1},\hat{h}_i^{k+1})\big)|\hat{h}_i^0\big] \\
&\leq \cfrac{\phi_{\max}}{1-\beta}\,\mathbb{E}^{\pi'}\big[\sum_{k=0}^{\infty}\beta^k d_b^k|
\hat{h}_i^0\big].
\label{eq:pfd2}
\end{aligned}
\end{equation}
\indent
Combining the above bounds, we obtain
\begin{equation}
\begin{aligned}
& V_{\phi}^{\pi'}(\hat{h}_i^0)-V_{\phi}^{\pi}(\hat{h}_i^0) \\
&\leq \mathbb{E}^{\pi'}\big[\sum_{k=0}^{\infty}\beta^kA_{\phi}^\pi(\hat{h}_i^k,u^k)|\hat{h}_i^0\big] + \cfrac{2\phi_{\max}}{1-\beta}\,\mathbb{E}^{\pi'}\big[\sum_{k=0}^{\infty}\beta^k d_b^k|\hat{h}_i^0\big].
\end{aligned}
\label{eq:lemma1-}
\end{equation}
\indent
Considering $\xi_i$, \eqref{eq:lemma1-} becomes $\eqref{eq:lemma1}$. This completes the proof.
\subsection{Proof of Lemma 2} Consider two policies $\pi'=(\pi'_i,\pi_{-i})$ and $\pi=(\pi_i,\pi_{-i})$. Using the definition of NE-gap, \eqref{eq:POMPG} and \eqref{eq:lemma1-}, we obtain \eqref{eq:lemma2}. This concludes the proof.
\subsection{Proof of Lemma 3}
Define mixed policy as
\[
\tilde{\pi}_{-i}(u_{-i}|\hat{h}_{-i})=\prod_{j=1}^{i-1}\pi_j^{t+1}(u_j|\hat{h}_j)\prod_{j=i+1}^{n}\pi_j^{t}(u_j|\hat{h}_j),\] 
and corresponding local advantage function as $\tilde{A}_{i,\phi}^t(\hat{h}_i,u_i)=\sum_{u_{-i}}\tilde{\pi}_{-i}(u_{-i}|\hat{h}_{-i})A_{\phi}^t(\hat{h}_i,u_i)$. The following term is obtained as 
\begin{equation}
\begin{aligned}
&\mathbb{E}_{\hat{h}_i,u}\big[ A_{\phi}^{\pi^t}(\hat{h}_i,u)\big] \\
&=\sum_{\hat{h}_i}d_{\xi}^{t+1}(\hat{h}_i)\,\sum_{i=1}^{n}\sum_{u_i}\big(\pi_i^{t+1}(u_i|
\hat{h}_i)-\pi_i^{t}(u_i|\hat{h}_i)\big)\tilde{A}_{i,\phi}^{\pi^t}(\hat{h}_i,u_i),
\label{eq:lemma3-1}
\end{aligned}
\end{equation}
due to $\pi^{t+1}-\pi^t=\sum_{i=1}^n (\pi_i^{t+1}-\pi_i^t)\,(\prod_{j<i}\pi_j^{t+1})\,(\prod_{j>i}\pi_j^{t})$, and the definition of $\tilde{A}_{i,\phi}^{\pi^t}(\hat{h}_i,u_i)$. \eqref{eq:lemma3-1} can be rewritten as
\begin{equation}
\begin{aligned}
&\sum_{u_i}\big(\pi_i^{t+1}(u_i|\hat{h}_i)-\pi_i^{t}(u_i|\hat{h}_i)\big)\tilde{A}_{i,\phi}^{\pi^t}(\hat{h}_i,u_i)  \\
&=\sum_{u_i}\pi_i^{t+1}(u_i|\hat{h}_i)A_{i,\phi}^{\pi^t}(\hat{h}_i,u_i) \\
&+\sum_{u_i}\big(\pi_i^{t+1}(u_i|\hat{h}_i)-\pi_i^{t}(u_i|\hat{h}_i)\big)\big(\tilde{A}_{i,\phi}^{\pi^t}(\hat{h}_i,u_i)-A_{i,\phi}^{\pi^t}(\hat{h}_i,u_i)\big).
\label{eq:lemma3-2}
\end{aligned}
\end{equation}
\indent
Using the definition of KL divergence and NPG update, the first term in \eqref{eq:lemma3-2} becomes
\begin{equation}
\begin{aligned}
&\sum_{u_i}\pi_i^{t+1}(u_i|\hat{h}_i)A_{i,\phi}^{\pi^t}(\hat{h}_i,u_i)  \\
&=\cfrac{1-
\beta}{\eta}\big(\text{KL}\big(\pi_i^{t+1}(\cdot|
\hat{h}_i)\,||\,\pi_i^{t}(\cdot|\hat{h}_i)\big)+
\log \big( G_i^t(w,l_i,y_i)\big)\big).
\label{eq:lemma3-3}
\end{aligned}
\end{equation}
\indent
The difference in advantage functions in \eqref{eq:lemma3-2} is bounded as
\begin{equation}
\begin{aligned}
& \big| \tilde{A}_{i,\phi}^{\pi^t}(\hat{h}_i,u_i)-A_{i,\phi}^{\pi^t}(\hat{h}_i,u_i) \big| \\
&\leq \cfrac{\phi_{\max}}{1-\beta}\,\sum_{j=1}^{n} \big\| \pi_j^{t+1}(\cdot|\hat{h}_j)-\pi_j^{t}(\cdot|\hat{h}_j)\big\|_1.
\label{eq:lemma3-4}
\end{aligned}
\end{equation}
\indent
Substituting \eqref{eq:lemma3-4} into \eqref{eq:lemma3-2} and using the Pinsker's inequality, we obtain
\begin{equation}
\begin{aligned}
&\sum_{u_i}\big(\pi_i^{t+1}(u_i|\hat{h}_i)-\pi_i^{t}(u_i|\hat{h}_i)\big)\big(\tilde{A}_{i,\phi}^{\pi^t}(\hat{h}_i,u_i)-A_{i,\phi}^{\pi^t}(\hat{h}_i,u_i)\big) \\
&\leq \cfrac{2n\phi_{\max}}{1-\beta}\sum_{j=1}^{n} \text{KL}\bigg(\pi_i^{t+1}(\cdot|\hat{h}_j)\,\|\,\pi_i^{t}(\cdot|\hat{h}_j)\bigg),
\label{eq:lemma3-6}
\end{aligned}
\end{equation}
and this completes the proof.

\subsection{Proof of Lemma 4}
Let $
u_i^*= \arg\max_{u_i} Q_i^{\pi}(\hat{h}_i,u_i) \in \mathcal{U}_i^*$, and $u_i^{-*} \in \mathcal{U}_i\,\backslash\,\mathcal{U}_i^*$. $G_i^t(\hat{h}_i)$ can be decomposed into the term with respect to $u_i^*$ and $u_i^{-*}$ as
\[
\begin{aligned}
G_i^t(\hat{h}_i) &=\sum_{u_i^*}\pi_i^t(u_i|\hat{h}_i)\exp\big(\cfrac{\eta\max_{u_i} A_i^{\pi^t}(\hat{h}_i,u_i)}{1-\beta}\big)\\
&\quad+\sum_{u_i^{-*}}\pi_i^t(u_i|\hat{h}_i)\exp\big(\cfrac{\eta A_i^{\pi^t}(\hat{h}_i,u_i)}{1-\beta}\big).
\end{aligned}
\]
\indent
Applying Taylor's expansion, we obtain the bound as
\[
\begin{aligned}
G_i^t(\hat{h}_i,u_i)
&\geq 1 + \cfrac{1}{2}\,\sum_{u_i^*}\pi_i^t(u_i|\hat{h}_i)\big(\cfrac{\eta\max_{u_i} A_i^{\pi^t}(\hat{h}_i,u_i)}{1-\beta}\big)^2 \\
&\geq 1+\cfrac{a}{2}\,\big(\cfrac{\eta\max_{u_i} A_i^{\pi^t}(\hat{h}_i,u_i)}{1-\beta}\big)^2.
\end{aligned}
\]
\indent
For $\eta \leq (1-\beta)^2$, we have $\cfrac{a}{2}\,\bigg(\cfrac{\eta\max_{u_i} A_i^{\pi^t}(\hat{h}_i,u_i)}{1-\beta}\bigg)^2 \leq \cfrac{1}{2}$, and use the inequality $\log (1+\lambda_3)\geq \dfrac{2}{3}\lambda_3$ for $0\leq \lambda_3 \leq \dfrac{1}{2}$, then the bound can be rewritten as

\begin{equation}
\log \big( G_i^t(\hat{h}_i,u_i) \big) \geq \cfrac{a}{3}\,\bigg(\cfrac{\eta\max_{u_i} A_i^{\pi^t}(\hat{h}_i,u_i)}{1-\beta}\bigg)^2.
\label{eq:lemma4o}
\end{equation}
\indent
Using \eqref{eq:lemma4o} and the definition of $M$, we obtain the bound of $\sum_{\hat{h}_i}d_{\xi}^{\pi^{t+1}}(\hat{h}_i)\sum_{i=1}^{n}\log\big(G_i^t(\hat{h}_i)\big)$ as \eqref{eq:lemma4}. This concludes the proof.


\ifCLASSOPTIONcaptionsoff
  \newpage
\fi

\end{document}